\documentclass[12pt]{extarticle}
\usepackage[utf8]{inputenc}

\usepackage{amsmath, amsthm, amssymb, color, url}
\usepackage{graphicx}

\usepackage{mathtools}
\usepackage{multirow}
\usepackage{booktabs}

\usepackage[shortlabels]{enumitem}

\usepackage[colorlinks=true, allcolors=blue]{hyperref} 

\usepackage{tikz}
\usepackage{pgfplots} 
\usepgfplotslibrary{fillbetween}
\pgfplotsset{compat=1.14}

\oddsidemargin \evensidemargin
\newcommand{\R}{\mathbb{R}}

\newcommand{\N}{\mathbb{N}}

\theoremstyle{definition}
\newtheorem{theorem}{Theorem}[section]
\newtheorem{corollary}[theorem]{Corollary}
\newtheorem{proposition}[theorem]{Proposition}

\theoremstyle{definition}
\newtheorem{defn}[theorem]{Definition}
\newtheorem{definition}[theorem]{Definition}
\newtheorem{example}[theorem]{Example}

\newtheorem{remark}[theorem]{Remark}
\newtheorem{proc}[theorem]{Procedure}

\usepackage{makecell}
\usepackage[version=4]{mhchem}

\newcommand{\scc}{\mathcal{S}}
\newcommand{\F}{f_{c, \kappa}}

\newcommand{\St}{{S}}
\newcommand{\Rnn}{\mathbb{R}_{\geq 0}}

\def\SS{\mathcal S}
\def\CC{\mathcal C}
\def\RR{\mathcal R}

\def\la{\leftarrow}

\newcommand{\lradot}{%
  \mathrel{\ooalign{\hfil$\vcenter{
   \hbox{$\mkern3mu\scriptscriptstyle\bullet$}}$\hfil\cr$\longleftrightarrow$\cr}
  }%
}

\DeclareMathOperator{\im}{im}

\newcommand{\vol}{\mathrm{Vol}}

\usepackage{longtable}

\newcommand{\defword}[1]{\textcolor{purple}{\underline{#1}}}

\usepackage[affil-it]{authblk}

\title{Mixed volume of small reaction networks}
\author[1]{Nida Obatake} 
\author[1]{Anne Shiu}
\author[2]{Dilruba Sofia}
\affil[1]{Department of Mathematics, Texas A\&M University} 
\affil[2]{Department of Mathematics, University of Massachusetts-Dartmouth}

\date{April 29, 2020}

\begin{document}

\maketitle

\begin{abstract}
An important invariant of a chemical reaction network is its maximum number of positive steady states.  This number, however, is in general difficult to compute.  Nonetheless, there is an upper bound on this number  -- namely, a network's mixed volume -- that is easy to compute.  Moreover, recent work has shown that, for certain biological signaling networks, the mixed volume does not 
greatly exceed 
the maximum number of positive steady states.  Continuing this line of research, we further investigate this overcount and also compute the mixed volumes of small networks, those with only a few species or reactions.
  \vskip 0.1cm
  \noindent {\bf MSC classes:} 52A39, 
    37C10, 
    12D10, 
    65H04, 
    80A30 
\end{abstract}

\section{Introduction} \label{sec:intro}
For chemical reaction networks, information about steady states -- both their number and their nature (stability, etc.) -- yields insight into a network's capacity for processing information.  
Therefore, there 
have been numerous investigations 
into the capacity for multiple steady states, especially for networks arising from biology (see, e.g.,~\cite{BP, CFMW, CTF06, DPST, mss-review,torres-feliu}).  

The next step, determining the maximum number of steady states of a given network, is more difficult.
Indeed, this question, mathematically, asks us to compute the maximum number of positive roots of a family of parametrized polynomial systems.  Therefore, we are interested in upper bounds on this maximum number that are easy to compute.  

One such bound, introduced in~\cite{OSTT}, is the mixed volume of a network (see also the closely related definition 
in~\cite{gross-hill}). 
This bound is surprisingly good for certain biological signaling networks, with the ``mixed-volume overcount'' -- the difference between the mixed volume and the maximum number of steady states -- no more than 2 or 4~\cite{OSTT}.  Related results for three infinite families of networks are obtained in~\cite{gross-hill}.

Here we further investigate the mixed volume and the mixed-volume overcount, with a focus on small networks, those with just a few species or reactions.  
Our results are as follows.  
First, for networks with only one species, we show how to read off the mixed volume (and mixed-volume overcount) directly from the network (Theorems~\ref{thm:mv-1species} and~\ref{thm:equatorial-networks}), and conclude that the mixed-volume overcount can be arbitrarily large (Corollary~\ref{cor:non-equatorial-networks}).  
Next, we investigate networks with two species and two reactions, and show that among those that are at-most-bimolecular, nearly all have mixed-volume overcount 0 (Theorem~\ref{thm:2-rxn-2-species}).
Thus, the mixed volume is an excellent bound for such networks.

The outline of our work is as follows. 
We provide background in Section~\ref{sec:background}.  
Our main results are presented in Section~\ref{sec:results}, and we end with a discussion in Section~\ref{sec:discussion}.

\section{Background} \label{sec:background}
Below, we give background on chemical reaction systems (Section~\ref{sec:CRS}), their steady states (Section~\ref{sec:steady-state}), mixed volume (Section~\ref{sec:MV}), and networks having only one species (Section~\ref{sec:1-species}).

\subsection{Chemical reaction systems} \label{sec:CRS}

\begin{definition} \label{def:crn}
A \defword{{reaction network}} $\defword{G}:=(\SS,\CC,\RR)$
consists of three finite sets:
(1) a set of \defword{{species}} $\defword{\SS} := \{A_1,A_2,\dots, A_s\}$; 
(2) a set  $\defword{\CC} := \{y_1, y_2, \dots, y_p\}$ of \defword{{complexes}} (finite nonnegative-integer combinations of the species); and 
 (3) a set of \defword{{reactions}}, which are ordered pairs of complexes, excluding diagonal pairs: $\defword{\RR} \subseteq  (\CC \times \CC) \smallsetminus \{ (y,y) \mid y \in \CC\}$.
\end{definition}

\noindent
Throughout our work, 
$s$ and $r$ denote the numbers of
species and reactions, respectively.  
A reaction network is \defword{genuine} if every species takes part in at least one reaction.

Writing the $i$-th complex as $y_{i1} A_1 + y_{i2} A_2 + \cdots + y_{is}A_s$ (here, $y_{ij} \in \mathbb{Z}_{\geq 0}$ 
is the \defword{{stoichiometric coefficient}} of $A_j$, for $j=1,2,\dots,s$), 
this complex is \defword{at-most-bimolecular} if $y_{i1}  + y_{i2}  + \dots + y_{is} \leq 2$.  A reaction network is \defword{at-most-bimolecular} if every complex in the network is at-most-bimolecular.

It is customary to write a reaction $(y_i,y_j)$ as $y_i \to y_j$, and $y_i$ is the called the \defword{reactant} and $y_j$ is the 
\defword{product}.
Also, a reaction $y_i \to y_j$ is \defword{{reversible}} if its  reverse reaction $y_j \to y_i$ is also in $\RR$, and we denote such a pair by $y_i \rightleftharpoons y_j$.
A reaction $y_i \to y_j$ 
defines the \defword{{reaction vector}}
 $\defword{y_j-y_i}$, which encodes the
net change in each species 
resulting from the reaction. The \defword{{stoichiometric matrix}} 
$\defword{\Gamma}$ is the $s \times r$ matrix whose $k$-th column 
is the reaction vector of the $k$-th reaction.  
Each reaction comes with a
\defword{{rate constant}} $\defword{\kappa_{ij}}$,
which is a positive parameter.

Next, 
we let $\defword{x_1},\defword{x_2},\ldots,\defword{x_s}$ represent the
concentrations of the $s$ species,
which we view as functions $\defword{x_i(t)}$ of time $t$.
Also, 
we define the monomial 
$\defword{\mathbf{x}^{y_i}} \,\,\, := \,\,\, x_1^{y_{i1}} x_2^{y_{i2}} \cdots  x_s^{y_{is}}~. $

A \defword{{chemical reaction system}} is the 
dynamical system that arises, via mass-action kinetics, from a chemical reaction
network $(\SS, \CC, \RR)$ and a choice of rate constants $(\kappa^*_{ij}) \in
\mathbb{R}^{r}_{>0}$ (recall that $r$ is the number of
reactions), as follows:
\begin{align} \label{eq:ODE-mass-action}
\frac{d\mathbf{x}}{dt} \quad = \quad \sum_{ y_i \to y_j~ {\rm is~in~} \RR} \kappa_{ij} \mathbf{x}^{y_i}(y_j - y_i) \quad =: \quad \defword{f_{\kappa}(\mathbf{x})}~.
\end{align}

Viewing the rate constants as a vector of parameters $\kappa=(\kappa_1, \kappa_2, \dots, \kappa_m)$, we have polynomials $f_{\kappa,i} \in \mathbb Q[\kappa,x]$, for $i=1,2, \dots, s$.  For simplicity, we will write $f_i$ rather than  $f_{\kappa,i}$.

The \defword{{stoichiometric subspace}}, 
  $\defword{\St}~:=~ {\rm span} \left( \{ y_j-y_i \mid  y_i \to y_j~ {\rm is~in~} \RR \} \right)$, 
 is the vector subspace of
$\mathbb{R}^s$ spanned by all reaction vectors
$y_j-y_i$.  
Thus, $\St = \im(\Gamma)$, where $\Gamma$ is the stoichiometric matrix.
Let $d=s-{\rm rank}(\Gamma)$.  
A \defword{conservation-law matrix} of $G$, denoted by $W$, is a row-reduced $d\times s$-matrix whose rows form a basis of 
the orthogonal complement of $S$.

A trajectory $x(t)$ that starts at a 
    positive vector $x(0)=x^0 \in
    \mathbb{R}^s_{> 0}$ 
remains, for all positive time,
 in the following \defword{stoichiometric compatibility class} with respect to the \defword{total-constant vector} $c\coloneqq W x^0 \in {\mathbb R}^d$: 
    \begin{align} \label{eqn:invtPoly}
    \scc_c~\coloneqq~ \{x\in {\mathbb R}_{\geq 0}^s \mid Wx=c\}~.
    \end{align}


\begin{example}\label{ex:2s2r-network}
Consider the network $G = \{ \ce{2A ->[$k_1$] 2B}~,~\ce{B ->[$k_2$] A} \}$. This network has $r=2$ non-reversible reactions
involving $p=4$ distinct complexes -- represented as vectors $
(2,0),~(0,2),~(0,1),~(1,0) 
$
-- 
on $s=2$ species, $A$ and $B$. 
Also, the network is genuine and at-most-bimolecular.
The stochiometric matrix of $G$ is 
\[
\Gamma=
\begin{bmatrix}
-2 & 1\\
2 & -1
\end{bmatrix}~.
\]
The stoichiometric subspace $S$, which has dimension $d=1$,
is spanned by $(1,-1)^{\text{T}}$, and a conservation-law matrix of $G$ is 
$
W=\begin{bmatrix}
1 & 1
\end{bmatrix}
$.
Let $\textbf{x}(t) = (x_1(t),x_2(t))\in \Rnn^2$
denote the vector of concentrations of species $A$ and $B$. 
A conservation law for $G$ is
$x_1+x_2=c_1$ for $c_1\in \Rnn$.
The chemical reaction system of $G$ arising from mass-action kinetics is
\begin{align*}
    \frac{d \textbf{x}}{dt} \quad &= \quad \begin{pmatrix}
     -2k_1x_1^2+k_2x_2 \\
     2k_1x_1^2-k_2x_2
    \end{pmatrix}~.
\end{align*}
\end{example}

\subsection{Steady states} \label{sec:steady-state}
For a chemical reaction system,
a \defword{{steady state}} is a nonnegative concentration vector $\defword{\mathbf{x}^*} \in \Rnn^s$ at which the right-hand side of the ODEs~\eqref{eq:ODE-mass-action}  vanish: $f_{\kappa} (\mathbf{x}^*) = 0$.  
We will focus on \defword{{positive steady states}} $\mathbf{x} ^* \in \mathbb{R}^s_{> 0}$.

To analyze steady states in a stoichiometric compatibility class, we 
use conservation laws 
in place of linearly dependent steady-state equations, as follows.
Let $I = \{i_1 < i_2< \dots < i_d\}$ denote the indices of the first nonzero coordinate of the rows of conservation-law matrix $W$.
For a total-constant vector $c$,
define the function $\F: {\mathbb R}_{\geq 0}^s\rightarrow {\mathbb R}^s$ as follows:
\begin{equation}\label{consys}
f_{c,\kappa,i} =\F(x)_i :=
\begin{cases}
f_{i}(x)&~\text{if}~i\not\in I,\\
(Wx-c)_k &~\text{if}~i~=~i_k\in I .
\end{cases}
\end{equation}
The system~\eqref{consys} is called the system \defword{augmented by conservation laws}. 

\begin{remark}
For networks without conservation laws, the augmented system is just the original system $\F$ in~(\ref{eq:ODE-mass-action}).
\end{remark}

\begin{definition}~ \label{def:mss}
\begin{enumerate}
\item A network is \defword{{multistationary}}
if there exist positive rate constants $\kappa_{ij}$
such that, for the corresponding chemical system~\eqref{eq:ODE-mass-action}, there is some stoichiometric compatibility class~\eqref{eqn:invtPoly} having two or more positive steady states. 

\item 
A network \defword{{admits $k$ positive steady states}}
(for some $k \in \mathbb{Z}_{\geq 0}$)
if there exists a choice of positive rate constants so that the resulting mass-action system 
has exactly $k$ positive steady states in some stoichiometric compatibility class.
\end{enumerate}

\end{definition}
The \defword{maximum number of positive steady states} 
of a network $G$ is the maximum value of~$k$ (with $k \in \mathbb{Z}_{\geq 0}$) for which $G$ admits $k$ positive steady states.

\subsection{Mixed volume} \label{sec:MV}
Here we recall from~\cite{OSTT} the mixed volume of a network, which is in general an upper bound on the maximum number of positive steady states.
For background on convex and polyhedral geometry, see~\cite{ewald,ziegler}. In particular, for 
a polynomial $f = b_1 x^{\sigma_1} + b_2 x^{\sigma_2} + \dots + b_{\ell} x^{\sigma_{\ell}} \in \mathbb{R}[x_1,x_2,\dots, x_s]$, where 
the 
exponent vectors $\sigma_i \in \mathbb{Z}^s$ are distinct and 
$b_i \neq 0$ for all $i$, the \defword{Newton polytope} of $f$
is the convex hull of its exponent vectors:
$	{\rm Newt}(f) \coloneqq {\rm conv} \{\sigma_1,~ \sigma_2,~ \dots~ ,~ \sigma_{\ell}\} ~\subseteq~ \mathbb{R}^s.$

\begin{defn} 
Let $P_1, P_2, \ldots,P_s\subseteq \R^s$ be polytopes. 
The volume of the Minkowski sum $\lambda_1P_1+ \lambda_2 P_2+ \ldots+\lambda_s P_s$ is a
degree-$s$ homogeneous polynomial in nonnegative variables $\lambda_1,\lambda_2,\ldots,\lambda_s$. In this polynomial, the coefficient 
of $\lambda_1 \lambda_2\cdots\lambda_s$, 
denoted by $\vol(P_1, P_2, \ldots,P_s)$, is the \defword{mixed volume} of $P_1,P_2,...,P_s$. 
\end{defn}
\noindent



\begin{defn}\label{def:mv-crn}
Let $G$ be a network with $s$ species, 
$r$ reactions, and 
a $d \times s$ conservation-law matrix $W$.
Let $f_{c,\kappa}$, as in~\eqref{consys}, denote the resulting system augmented by conservation laws. Let $c^*\in \mathbb{R}^d_{\neq 0}$, and let 
$\kappa^* \in \mathbb{R}^r_{>0}$ be generic. 
Let $P_1, P_2, \ldots,P_s\subset \R^s$ be the Newton polytopes of $f_{c^*,\kappa^*,1},f_{c^*,\kappa^*,2}, \ldots, f_{c^*,\kappa^*,s}$, respectively. The \defword{mixed volume of $G$} (with respect to $W$)
is the mixed volume of $P_1,P_2,\ldots,P_s$.
\end{defn}

The mixed volume (Definition~\ref{def:mv-crn}) is well defined~\cite[Remark 8]{OSTT}.  The next result follows from Bernstein's theorem~\cite{bernstein} (see \cite[Proposition 8]{OSTT}): 

\begin{proposition} \label{prop:bounds}
For every network, the following inequality relates
the
maximum number of positive steady states  
and the mixed volume (with respect to any conservation-law matrix):
\begin{align} \label{eq:MV-bound}
\mathrm{maximum~number~of~positive~steady~states}
~\leq ~ 
{\rm mixed~volume}~.
\end{align}
\end{proposition}

The \emph{mixed-volume overcount} 
measures how tight the bound~\eqref{eq:MV-bound} is.
Of particular interest are networks with 0 mixed-volume overcount, 
because for these networks, the mixed volume precisely and efficiently calculates the maximum number of positive steady states.

\begin{definition}\label{def:overcount}
The \defword{mixed-volume overcount} of a reaction network $G$ is 
\[
~({\rm mixed~volume~of~}G)
~- ~ 
(\mathrm{maximum~number~of~positive~steady~states~of~}G)
~.
\]
\end{definition}

An example considered in earlier work is the extracellular signal-regulated kinase (ERK) network. 
This is an important biological signaling network known to be multistationary (and also bistable) \cite{OSTT, long-term}. 
For the ERK network and several simplified versions of the network, the mixed-volume overcount is 2 -- for the fully irreversible and reduced subnetworks -- or (conjectured to be) 4 - for the full network and the subnetwork obtained by removing one reaction (specifically, the reaction  $k_{\rm{on}}$)~\cite[Proposition 9 and Conjecture 1]{OSTT}.

\subsection{One-species networks} \label{sec:1-species}
Here we recall some definitions from~\cite{which-small}.

\begin{definition} \label{def:arrow-diagram}
Let $G$ be a reaction network containing only one species $A$. Each reaction of $G$ therefore has the form $aA \to bA$, where $a,b \ge 0$ and $a \ne b$. Let $m$ be the number of (distinct) reactant complexes, and let $a_1< a_2 < \ldots < a_m$ be the stoichiometric coefficients. The \defword{{arrow diagram}} of $G$, denoted by $\defword{\rho} = (\rho_1, \ldots , \rho_m)$, is the element of $\{\to , \la, \lradot \}^m$ with:
\begin{equation*}
 \defword{\rho_i}~:=~ 
 \left\lbrace\begin{array}{ll}
   \to & \text{if for all reactions $a_iA \to bA$ in $G$, we have $b > a_i$} \\
   \la & \text{if for all reactions $a_iA \to bA$ in $G$, we have $b < a_i$} \\
   \lradot & \text{otherwise.}
 \end{array}\right.
\end{equation*}
\end{definition}

\begin{definition} \label{def:2-sign-change}
For nonnegative integers $T \geq 0$, a \defword{{$T$-alternating network}} is a 1-species network with 
exactly $T+1$ reactions and with arrow diagram $\rho \in \{\to , \la\}^{T+1}$ such that, if $T \geq 1$, we have $\rho_i = \to$ if and only if $\rho_{i+1} = \la$ for all $i \in \{1,2, \ldots, T\}$. 
\end{definition}


\begin{example} \label{ex:alternating} 
Consider the following network:
 	\begin{align*}
	G~=~\{ 0 \leftarrow A \to 2A \rightleftharpoons 3A\}~.
	\end{align*}
 Two 1-alternating subnetworks of $G$ 
	have arrow diagram $(\to, \la)$:
	$\{ A \to 2A,~ 2A \leftarrow 3A\}$ and
	$\{  2A \to 3A,~ 2A \leftarrow 3A\}$.
On the other hand, $\{ 0 \leftarrow A,~ A \to 2A\}$ is {\em not} a 1-alternating subnetwork of $G$: 
	its arrow diagram is $(\lradot)$. 
Finally, $\{ 0 \leftarrow A,~ 2A \to 3A, ~ 2A \leftarrow 3A \}$ is a 2-alternating subnetwork of $G$ with arrow diagram $(\la, \to, \la)$.
\end{example}

The following result follows directly from~\cite[Theorem 3.6]{which-small} and its proof:
\begin{proposition}[Number of steady states for one-species networks] \label{prop:joshi-shiu}
Let $G$ be a reaction network with only one species (and at least one reaction).  Then, the maximum number of positive steady states of $G$
equals the maximum value of $T \in \mathbb{Z}_{\geq 0}$ for which $G$ has a $T$-alternating subnetwork.
\end{proposition}

\section{Results} \label{sec:results}
In Section~\ref{sec:1-rxn-or-1-species}, we characterize the mixed volume and mixed-volume overcount of networks with only one reaction or one species.  As a consequence, we show that the mixed-volume overcount can be arbitrarily large (Corollary~\ref{cor:non-equatorial-networks}).
Subsequently, in Section~\ref{sec:2-species-2-rxns}, we show that nearly all (genuine) networks with two species and two reactions 
have mixed-volume overcount 0
(Theorem~\ref{thm:2-rxn-2-species}).

\subsection{Networks with only one reaction or one species} \label{sec:1-rxn-or-1-species}

\begin{proposition}[Mixed volume of one-reaction networks]\label{prop:mv-1reaction}
For a network with only a single reaction, the mixed volume is 0 and the mixed-volume overcount is 0.
\end{proposition}

\begin{proof}
Let $G$ be a network with only one reaction.
The right-hand side of the ODE consists of a single monomial, 
so the Newton polytope is just a point (the exponent vector of the monomial).
Hence, the mixed volume of $G$ is 0, 
and so the mixed-volume overcount is 0, by Proposition~\ref{prop:bounds}.
\end{proof}

\begin{theorem}[Mixed volume of one-species networks] \label{thm:mv-1species}
Let $G$ be a reaction network that contains only one species $A$. 
Let $m$ be the number of (distinct) reactant complexes, and let $a_1< a_2 < \ldots < a_m$ be their stoichiometric coefficients. Then 
\begin{align*}
\mathrm{mixed~volume~of~} G
    ~=~
    a_m-a_1~.
\end{align*}
\end{theorem}

\begin{proof}
As $G$ has only one species, 
there are no conservation laws
and only one differential equation. 
In this equation, the leading monomial is $x_1^{a_m}$, and the lowest-degree monomial is $x_1^{a_1}$.  
The Newton polytope of this single polynomial is therefore the line segment between $a_1$ and $a_m$. Thus, by definition, the mixed volume of $G$ is $a_m-a_1$.
\end{proof}

\begin{corollary}\label{cor:non-equatorial-networks}
The mixed-volume overcount can be arbitrarily large.
\end{corollary}

\begin{proof}
Consider the network $\ce{0 <=>[$k_1$][$k_2$] n A}$, 
where $n\in \N$. 
The right-hand side of the mass-action ODEs~\eqref{eq:ODE-mass-action} 
is the polynomial~$-k_2 a^n + k_1$, which 
has
precisely one positive real root (namely, $a=\sqrt[n]{k_1/k_2}$).  However, by Theorem~\ref{thm:mv-1species}, the mixed volume is $n$. So, the mixed-volume overcount is $(n-1)$.  
\end{proof}

\begin{theorem}[One-species networks with mixed-volume overcount 0]\label{thm:equatorial-networks}
Let $G$ be a reaction network that contains only one species $A$.  
Let $m$ be the number of (distinct) reactant complexes, and let $a_1< a_2 < \ldots < a_m$ be their stoichiometric coefficients. Then $G$
has mixed-volume overcount 0
if and only if 
$G$ has an $(m-1)$-alternating subnetwork
and 
$a_i=a_1+i-1$ for all $i \in \{2, 3, \dots,m\}$.
\end{theorem}

\begin{proof}
This result follows directly from Proposition~\ref{prop:joshi-shiu} and Theorem~\ref{thm:mv-1species}.
\end{proof}

\begin{example}[Example~\ref{ex:alternating} continued] 
By Theorem~\ref{thm:equatorial-networks}, the network from Example~\ref{ex:alternating}
has mixed-volume overcount 0.
Indeed, it is a one-species network with 3 distinct reactant complexes (note that $\ce{0}$ is not a reactant complex in this network) satisfying $a_i=a_1+i-1$ for $i\in \{2,3\}$ (here the notation is as in Theorem~\ref{thm:equatorial-networks} with $a_1=1$), and it has a $2$-alternating subnetwork.
\end{example}

\subsection{Networks with two species and two reactions} \label{sec:2-species-2-rxns}
Up to relabeling species, there are 210 genuine, at-most-bimolecular networks with two species and two reactions~\cite{banaji-count}. 
These networks, which were enumerated by Banaji, are listed at \url{https://reaction-networks.net/networks/}. 
Here we determine that 
92\% of
these networks
have mixed-volume overcount 0 (Theorem~\ref{thm:2-rxn-2-species}); the 16 exceptional networks are listed in Table~\ref{tab:16}.

The following result, which follows directly from~\cite[Lemma 2.7, Lemma 4.1, and Theorem~4.8]{which-small} 
(also cf.~\cite[Corollary~4.12 and the preceding paragraph]{which-small}), 
implies that the 210 networks we consider in this subsection
 are {\em not} multistationary.

\begin{proposition}  \label{prop:at-most-bimol}
If $G$ is an at-most-bimolecular reaction network with exactly two species and two reactions, then the maximum number of positive steady states of $G$ is at most 1. 
Moreover, this maximum number is 1 if the two reaction vectors of $G$ are negative scalar multiples of each other, and 0 otherwise.
\end{proposition}

Proposition~\ref{prop:at-most-bimol} and the definition of mixed-volume overcount directly yield the following:
\begin{corollary} \label{cor:at-most-bimol}
Let $G$ be an at-most-bimolecular reaction network with exactly two species and two reactions.  If the mixed volume of $G$ is at least 2, then the mixed-volume overcount is at least 1.
\end{corollary}

We use the following procedure to compute (by using {\tt PHCpack}~\cite{phcpack}, as in~\cite{OSTT}) 
the mixed-volume overcount of a 2-species, 2-reaction network:  

\begin{proc}\label{proc:mv-overcount} 
~
\emph{Input}: A 2-species, 2-reaction network  $G$.

\emph{Output}: the mixed-volume overcount of $G$.
\begin{enumerate}\setcounter{enumi}{-1} 
    \item Compute the system augmented by conservation laws~(\ref{consys}), denoted by $\F$, for some choice of conservation-law matrix $W$. 
    \item Compute the mixed volume of $G$, as follows.
    Viewing the two polynomials in $\F$ as polynomials in $x_1$ and $x_2$, 
    substitute 1 for all coefficients; let {\tt poly1} and {\tt poly2} be the 
    resulting polynomials.  
    Next, run the following {\tt Macaulay2} code: 
\begin{verbatim}
    loadPackage "PHCpack"
    S = CC[x1,x2];
    F = {poly1 , poly2};
    mixedVolume(F)
\end{verbatim}
    \item Compute the maximum number of positive steady states:
    \begin{enumerate}
        \item 
        If $G$ has no linear conservation laws, the maximum number of positive steady states is 0.
    
        \item 
        If $G$ has a linear conservation law, determine the maximum number of positive steady states of $G$ by analyzing the possible numbers of positive roots of $f_{c,\kappa}=0$ 
        (or by other means, e.g., if applicable, Proposition~\ref{prop:at-most-bimol}). 
        


        
    \end{enumerate}
    \item 
    Output 
    the difference between the mixed volume (from Step 1) and the maximum number of positive steady states (from Step 2).
\end{enumerate}
\end{proc}


\begin{proof}[Proof of correctness of Procedure~\ref{proc:mv-overcount}]

    The correctness of Step 1 is due to the fact that mixed volume considers only the supports of polynomials.  
    The correctness of Step 2(a) follows from \cite[Lemma 4.1]{which-small}.  Step 2(b) is correct by construction of $\F$. 
    Finally, the correctness of Step~3 follows directly from the definition of mixed-volume overcount~(Definition~\ref{def:overcount}).  
\end{proof}

\begin{example}\label{ex:max-0pos-ss}
Consider $G=$ \{\ce{A + B ->[$k_1$] 2B <-[$k_2$] 2A}\}. 

\begin{enumerate}\setcounter{enumi}{-1}
    \item The system augmented by conservation laws is 
 \begin{equation}\label{eq:max-0pos-ss-aug-sys}
         \begin{cases}
         f_1(x_1,x_2) = x_1+x_2 - c_1 \\
         f_2(x_1,x_2) = 2k_2x_1^2+k_1x_1x_2~~~.
     \end{cases}
 \end{equation}
 
    \item Take 
    $k_1=2k_2=-c_1=1$
    in (\ref{eq:max-0pos-ss-aug-sys}), and compute the mixed volume of the resulting polynomial system. The mixed volume of the network is 1.
    \item We compute the maximum number of steady states:
    \begin{enumerate}
        \item There is a linear conservation law (namely, $f_1$), so continue to Step 2(b).

        \item The reaction vectors, $(-1,1)$ and $(-2, 2)$, are {\em not} negative scalar multiples of each other.  So, by Proposition~\ref{prop:at-most-bimol}, the maximum number of positive steady states is~0.  Alternatively, notice that $f_2(x_1^*,x_2^*)>0 $ when $x_1^*,x_2^*>0$, and so $\F=0$ never has positive roots. 
    \end{enumerate}
    \item The mixed-volume overcount is $1-0=1$.
\end{enumerate}

\end{example}


Next we provide two more examples of genuine 2-species, 2-reaction networks.
These examples show that determining the maximum number of positive steady states by analyzing the roots of $\F=0$ (Step 2(b) of Procedure~\ref{proc:mv-overcount}) is not straightforward in general.

\begin{example}[Example~\ref{ex:2s2r-network} continued]\label{ex:2s2r-network-continued}
Recall the genuine 2-species, 2-reaction network 
$ \{ \ce{2A ->[$k_1$] 2B},~\ce{B ->[$k_2$] A} \}$. 
Using Procedure~\ref{proc:mv-overcount}, we show below that the mixed-volume overcount of the network is 1.

\begin{enumerate}\setcounter{enumi}{-1}
    \item The system augmented by conservation laws is 
 \begin{equation}\label{eq:2s2r-aug-sys}
         \begin{cases}
         f_1(x_1,x_2) = x_1+x_2 - c_1 \\
         f_2(x_1,x_2) = 2k_1x_1^2-k_2x_2~~~.
     \end{cases}
 \end{equation}
 
    \item Take
    $2k_1=-k_2=-c_1=1$
    in (\ref{eq:2s2r-aug-sys}), 
    and compute the mixed volume of the resulting polynomial system. The mixed volume of the network is 2.
    \item We compute the maximum number of steady states:
    \begin{enumerate}
        \item There is a linear conservation law (namely, $f_1$), so continue to Step 2(b).

        \item 
        The reaction vectors are $(-2,2)$ and $(1,-1)$, which are negative scalar multiples of each other.  So, by Proposition~\ref{prop:at-most-bimol}, the maximum number of positive steady states is 1.  Alternatively, we  analyze the roots of $\F=0$, as follows.
        First, 
        $f_1=0$ yields $x_2=c_1-x_1$, which we substitute into $f_2=0$
        to get
    \[g(x_1)~=~2k_1x_1^2-k_2(c_1-x_1)~=~ 2k_1x_1^2+k_2x_1-k_2c_1~.\]
    This is a quadratic in $x_1$ with positive leading coefficient and negative vertical intercept (since $k_1,k_2,c_1>0$).
    Thus, for every choice of $k_1,k_2,c_1>0$, the quadratic has a unique positive real root in $x_1$, namely, $
    x_1^* = \left(-k_2+\sqrt{k_2^2+8c_1k_1k_2}\right)/(4k_1)$. 
    Therefore, the maximum number of steady states is at most 1.  In fact, this number is 1: when $k_1=1/2$, $k_2=1$ and $c_1=2$, there is a unique positive steady state, namely, $(x_1^*,x_2^*) = (1, 1)$.
    \end{enumerate}
    \item The mixed-volume overcount is $2-1=1$.
\end{enumerate}
\end{example}

\begin{example}\label{ex:network22}
Let $G =$  \{\ce{2A ->[$k_1$] 2B ->[$k_2$] A + B}\}. 

\begin{enumerate}\setcounter{enumi}{-1}
    \item The system augmented by conservation laws is 
 \begin{equation}\label{eq:network22-aug-sys}
         \begin{cases}
         f_1(x_1,x_2) = x_1+x_2 - c_1 \\
         f_2(x_1,x_2) = 2 k_1x_1^2-k_2x_2^2~~~.
     \end{cases}
 \end{equation}
 
    \item Take $2k_1=-k_2=-c_1=1$ in (\ref{eq:2s2r-aug-sys}), and compute the mixed volume of the resulting polynomial system. The mixed volume of the network is 2.
    \item We compute the maximum number of steady states:
    \begin{enumerate}
        \item There is a linear conservation law (namely, $f_1$), so continue to Step 2(b).

        \item 
        The reaction vectors, $(-2,2)$ and $(1,-1)$, are negative scalar multiples of each other.  So, Proposition~\ref{prop:at-most-bimol} implies that the maximum number of positive steady states is 1.  
        An alternate approach is as follows.
        We solve $f_2=0$ for $x_2$ (and use the fact that we are interested in only positive $x_1,x_2$), which yields $x^*_2=(\sqrt{2k_1/k_2}) x^*_1$. Next, we substitute this expression into $f_1=0$ and then solve to obtain $x^*_1=c_1/(1+\sqrt{2k_1/k_2})$.  Thus, 
        the network always admits 
        a unique positive steady state $(x_1^*,x_2^*)$.
    \end{enumerate}
    \item The mixed-volume overcount is $2-1=1$.
\end{enumerate}

\end{example}

\begin{remark}\label{rmk:counting-pos-roots-of-g}
The approaches that we present in this section 
for computing the maximum number of steady states of a network (Steps 2(a) and 2(b) of Procedure~\ref{proc:mv-overcount}) 
rely on the fact that the networks are 
at-most-bimolecular and have only two reactions and two species.
In general, however, completing Step 2 is not straightforward:
as mentioned in the Introduction, it requires  
counting the number of positive real roots of a parametrized polynomial system. 
This complication further motivates the need for graphical, algebraic, and geometric tools for counting positive steady states, in order to bypass a direct analysis of the polynomial system $f_{c,\kappa} = 0$.

\end{remark}

By applying Procedure~\ref{proc:mv-overcount}, we obtain a classification of genuine, at-most-bimolecular networks with two species and two reactions (Theorem~\ref{thm:2-rxn-2-species}).


\begin{table}[ht!]
    \centering
    \begin{tabular}{ccc}
    \hline
    & \textbf{Network} & \textbf{Mixed volume}  \\ \hline
    $(1)~$ & $\ce{2A -> 2B -> A + B}$ & 2\\ \hline
    $(2)~$ & 
    $\ce{2A -> 2B\ ,\ B -> A}$ & 2\\ \hline
    $(3)~$ & 
    $\ce{2A -> A}~,~\ce{ B -> A + B}$ & 2\\ \hline
    $(4)~$ & 
    $\ce{B -> A}~,~\ce{2A -> A + B}$ & 2\\ \hline
    $(5)~$ & 
    $\ce{B -> A}~,~\ce{2B -> A + B}$ & 1 \\ \hline
    $(6)~$ & 
    $\ce{2A <=> 2B}$ & 2\\ \hline
    $(7)~$ & 
    $\ce{2A -> A + B <- 2B}$ & 2\\ \hline
    $(8)~$ & 
    $\ce{B -> A}~,~\ce{2B -> 2A}$ & 1 \\ \hline
    $(9)~$ & 
    $\ce{B -> 2B}~,~\ce{A -> A + B}$ & 1 \\ \hline
    $(10)$ & 
    $\ce{2B -> 0}~,~\ce{A -> A + B}$ & 2\\ \hline
    $(11)$ & 
    $\ce{A <=> 2B}$ & 2 \\ \hline
    $(12)$ & 
    $\ce{A + B -> 2B <- 2A}$ & 1 \\ \hline
    $(13)$ & 
    $\ce{2A -> A + B -> 2B}$ & 1 \\ \hline
    $(14)$ & 
    $\ce{2A -> A}~,~\ce{A + B -> B}$ & 1 \\ \hline
    $(15)$ & 
    $\ce{A + B <=> 0}$ & 2 \\ \hline
    $(16)$ & 
    $\ce{B -> A}~,~\ce{A + B -> 2A}$ & 1 \\ \hline
    \end{tabular}
    \caption{Genuine, at-most-bimolecular networks with two species and two reactions for which the mixed-volume overcount is nonzero.  Each network has mixed-volume overcount 1. \label{tab:16}}
\end{table}

\begin{theorem}[Mixed volume of two-species, two-reaction networks]
\label{thm:2-rxn-2-species}
Let $G$ be a genuine, at-most-bimolecular network with 2 species and 2 reactions.  Then $G$ has mixed-volume overcount 0 if and only if $G$ is  (up to relabeling species) {\em not} one of the 16 networks listed in Table~\ref{tab:16}.  Moreover, each network in Table~\ref{tab:16} has mixed-volume overcount 1.
\end{theorem}

\begin{proof}
Using Procedure~\ref{proc:mv-overcount},
we computed the mixed-volume overcount for all genuine 2-species,~2-reaction networks; see the supplementary file {\tt MV-overcount-2s-2r-networks.csv} in the repository \url{https://github.com/neeedz/mixedvolume}.
More details are as follows. 
Among the 210 networks, 
    185 of them have mixed volume 0
    and thus have mixed-volume overcount 0.
 For the remaining 25 networks (see Appendix~\ref{app:code}), 
 it is straightforward to compute the maximum number of positive steady states using 
 Proposition~\ref{prop:at-most-bimol} or by directly analyzing the system $f_{c,\kappa}=0$ as in Examples~\ref{ex:max-0pos-ss}--\ref{ex:network22}.  
 %
\end{proof}

We end this section by investigating why the networks in Table~\ref{tab:16} have nonzero mixed-volume overcount.  
These 16 networks fall into four classes:
\begin{enumerate}
    \item Networks (3), (9), (10), and (14) are essentially one-species networks (for each network, one of the two ODEs is 0), and so can be analyzed using the results in Section~\ref{sec:1-rxn-or-1-species}.
    \item Networks (6), (11), and (15) consist of a single pair of reversible reactions, so 
    (e.g., by Proposition~\ref{prop:at-most-bimol})
    the maximum number of positive steady states is 1.
    \item Networks (5), (8), (12), (13), and (16) have one species that is consumed in every reaction (while the other species is produced).  Thus, the maximum number of positive steady states is 0.
    \item Networks (1), (2), (4), and (7) (and also networks (3), (6), (10), (11), 
    and (15)) have mixed volume 2, so, by Corollary~\ref{cor:at-most-bimol}, the mixed-volume overcount is at least 1.
\end{enumerate}

\begin{remark}\label{rmk:generalizing-procedure-for-mvo}
In 
Examples~\ref{ex:2s2r-network-continued} and~\ref{ex:network22}, we computed the maximum number of positive steady states (Step 2 of Procedure~\ref{proc:mv-overcount}) by reducing the system $\F=0$ to a single univariate polynomial, and then checking that the positive roots (which can be viewed as ``partial solutions'') can be extended to positive roots of the original system.
Doing this for general networks, however, is difficult.  Indeed, for readers with knowledge of algebraic geometry, we note that the Extension Theorem~\cite[pp.~118-120]{CLO} requires an algebraically closed field and polynomials with a certain shape. 
%
\end{remark}

\begin{example}
Consider the following network with 3 species and 10 reactions:
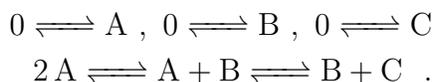
\begin{align*}
    \ce{0 <=> A}~,~\ce{0 <=> B}~,~\ce{0 <=> C}\\
    \ce{2A <=> A + B <=> B + C}~~.
\end{align*}

\noindent This network has no conservation laws, and its augmented system is
\[
\begin{cases}
f_1 = k_1 - k_2 x_1 - k_7 x_1^2 + (k_8 - k_9) x_1 x_2 + k_{10} x_2 x_3\\
f_2 = k_3 - k_4 x_2 + k_7 x_1^2 - k_8  x_1 x_2\\
f_3 = k_5 - k_6 x_3 + k_9 x_1 x_2 - k_{10} x_2 x_3~~.
\end{cases}
\]
Analyzing the augmented system is challenging, and determining the maximum number of steady states of the network is not straightforward. 
This number is at least 2~\cite{mss-review},  
and we compute that its mixed volume is 6.
What is the mixed-volume overcount? 
Our wish is to answer this question in the future 
through a generalized version of Procedure~\ref{proc:mv-overcount}.
\end{example}

\section{Discussion} \label{sec:discussion}
Recall that our interest in the mixed volume of a reaction network comes from the fact that it bounds the maximum number of positive steady states.  We saw in previous work that this bound is surprisingly good for certain signaling networks, and here we again found that this bound performs well for small networks that are at-most-bimolecular.  As networks arising in biological applications are typically at-most-bimolecular, we might expect the mixed-volume overcount to be low for biological networks of small to medium size.

Another future research direction pertains to one aim of this work, which is to read off the mixed volume directly from a network.  We now can do this for networks with just one reaction or one species (Section~\ref{sec:1-rxn-or-1-species}). As for at-most-bimolecular networks with two reactions and two species, the mixed volume is (with the exception of the 16 networks in Table~\ref{tab:16}) exactly the maximum number of positive steady states, which can be ascertained using results in~\cite{which-small}.  We would like similar results for networks with more reactions or more species.

Continuing this line of investigation, we ask, {\em How do operations on networks affect the mixed volume (and thus the mixed-volume overcount)?} For instance, in Table~\ref{tab:16}, networks (1) and (7) can be obtained from each other by ``stretching'' one reaction (without changing the reactant or reaction vector); and similarly for networks (2) and (4).  Moreover, this operation does not affect the mixed volume or the overcount.  
(This line of investigation therefore would be 
somewhat similar in spirit to 
the work of Rojas~\cite{rojas1994convex} and Bihan and Soprunov~\cite{bihan-soprunov}.) 
Indeed, having a list of operations and their effect on the mixed volume would greatly aid our classification of networks. 

\subsection*{Acknowledgements}
This research was initiated by DS in the 2019 REU in
the Department of Mathematics at Texas A\&M University, supported by the NSF
(DMS-1757872).  NO and AS were partially supported by the NSF (DMS-1752672).
We thank Taylor Brysiewicz for helpful discussions.
\bibliographystyle{plain}
\bibliography{mixedvol-crn.bib}

\appendix

\newpage

\section{Networks with nonzero mixed volume}\label{app:code}

Below, we list the 25 genuine 2-species, 2-reaction networks with nonzero mixed volume, together with their maximum number of positive steady states and their augmented systems.  The first 16 networks here coincide with those listed in Table~\ref{tab:16}.

\begin{longtable}{ccccl}
    \hline
    & \textbf{Network} & \textbf{Mixed volume}  &
    {\bf Max} $\#$ & \textbf{System}
\\ \hline
    
    $(1)~$ 
    & $\ce{2A -> 2B -> A + B}$ 
    & 2 & 1
    & $\begin{cases}
    a + b - c_1\\
    2k_1a^2 - k_2b^2
    \end{cases}$
    \\ \hline
    
    $(2)~$ 
    & 
    $\ce{2A -> 2B\ ,\ B -> A}$ & 2 & 1
    & $\begin{cases}
    a+b-c_1\\
    2k_1a^2-k_2b
    \end{cases}$
    \\ \hline
    
    $(3)~$ 
    & 
    $\ce{2A -> A}~,~\ce{ B -> A + B}$ 
    & 2 & 1
    & $\begin{cases}
    -k_1a^2+k_2b\\
    0
    \end{cases}$
    \\ \hline
    
    $(4)~$ 
    & 
    $\ce{B -> A}~,~\ce{2A -> A + B}$ 
    & 2 & 1
    & $\begin{cases}
    a+b-c_1\\
    k_2a^2-k_1b
    \end{cases}$
    \\ \hline
    
    $(5)~$ 
    & 
    $\ce{B -> A}~,~\ce{2B -> A + B}$ 
    & 1 & 0 
    & $\begin{cases}
    a+b-c_1\\
    -k_2b^2-k_1b
    \end{cases}$
    \\ \hline
    
    $(6)~$ 
    & 
    $\ce{2A <=> 2B}$ 
    & 2 & 1
    & $\begin{cases}
    a+b-c_1\\
    2k_1a^2-2k_2b^2
    \end{cases}$
    \\ \hline
    
    $(7)~$ 
    & 
    $\ce{2A -> A + B <- 2B}$ 
    & 2 & 1
    & $\begin{cases}
    a+b-c_1\\
    k_1a^2-k_2b^2
    \end{cases}$
    \\ \hline

    $(8)~$ 
    & 
    $\ce{B -> A}~,~\ce{2B -> 2A}$ 
    & 1 & 0
    & $\begin{cases}
    a+b-c_1\\
    -k_1b-2k_2b^2
    \end{cases}$
    \\ \hline
    
    $(9)~$ 
    & 
    $\ce{B -> 2B}~,~\ce{A -> A + B}$ 
    & 1 & 0
    &
    $\begin{cases}
    0\\
    k_1b+k_2a
    \end{cases}$
    \\ \hline
    
    $(10)$ 
    & 
    $\ce{2B -> 0}~,~\ce{A -> A + B}$ 
    & 2 & 1
    &
    $\begin{cases}
    0\\
    -2k_1b^2+k_2a
    \end{cases}$
    \\ \hline
    
    $(11)$ 
    & 
    $\ce{A <=> 2B}$ 
    & 2 & 1
    &
    $\begin{cases}
    a+b-c_1\\
    -k_2b^2+k_1a
    \end{cases}$
    \\ \hline

    $(12)$ 
    & 
    $\ce{A + B -> 2B <- 2A}$ 
    & 1 & 0
    &
    $\begin{cases}
    a+b-c_1\\
    2k_2a^2+k_1ab
    \end{cases}$
    \\ \hline
    
    $(13)$ 
    & 
    $\ce{2A -> A + B -> 2B}$ 
    & 1 & 0
    &
    $\begin{cases}
    a+b-c_1\\
    k_1a^2+k_2ab
    \end{cases}$
    \\ \hline

    $(14)$ 
    & 
    $\ce{2A -> A}~,~\ce{A + B -> B}$ 
    & 1 & 0
    &
    $\begin{cases}
    -k_2a^2-k_1ab\\
    0
    \end{cases}$
    \\ \hline

    $(15)$ 
    & 
    $\ce{A + B <=> 0}$ 
    & 2 & 1
    &
    $\begin{cases}
    -k_1ab+k_2\\
    a-b
    \end{cases}$
    \\ \hline

    $(16)$ 
    & 
    $\ce{B -> A}~,~\ce{A + B -> 2A}$ 
    & 1 & 0 
    &
    $\begin{cases}
    a+b-c_1\\
    -k_1b-k_2ab
    \end{cases}$
    \\ \hline

    $(17)~$ & 
    $\ce{0 -> 2B}~,~\ce{A + B -> A}$ 
    & 1 & 1
    &      
    $\begin{cases}
    0\\
    -k_2ab+2k_1
    \end{cases}$
    \\ \hline

    $(18)~$ 
    & 
    $\ce{2B -> 0}~,~\ce{A + B -> A}$ 
    & 1 & 1
    & 
    $\begin{cases}
    0\\
    -2k_1b^2-k_2ab\\
    \end{cases}$
    \\ \hline
    
    $(19)~$ 
    & 
    $\ce{A + B -> 2A -> 2B}$ 
    & 1 & 1
    &
    $\begin{cases}
    a+b-c_1\\
    2k_2a^2-k_1ab
    \end{cases}$
    \\ \hline

    $(20)~$ 
    & 
    $\ce{A + B -> 2B -> A + B}$ 
    & 1 & 1
    &
    $\begin{cases}
    a+b-c_1\\
    k_1ab-k_2b^2
    \end{cases}$
    \\ \hline
    
    $(21)~$ 
    & 
    $\ce{A + B -> 2B}~,~\ce{B -> A}$ 
    & 1 & 1 
    & 
    $\begin{cases}
    a+b-c_1\\
    k_1ab-k_2b
    \end{cases}$
    \\ \hline
    
    $(22)~$ 
    & 
    $\ce{A -> 0}~,~\ce{B -> A + B}$ 
    & 1 & 1
    &
    $\begin{cases}
    -k_1a+k_2b\\
    0
    \end{cases}$
    \\ \hline
    
    $(23)~$ 
    & 
    $\ce{A <=> B}$ 
    & 1 & 1
    &
    $\begin{cases}
    a+b-c_1\\
    k_1a-k_2b
    \end{cases}$
    \\ \hline

    $(24)~$ 
    & 
    $\ce{A + B -> A}~,~\ce{0 -> B}$ 
    & 1 & 1 
    & 
    $\begin{cases}
    0\\
    -k_1ab+k_2
    \end{cases}$
    \\ \hline

    $(25)~$ 
    & 
    $\ce{A + B <=> A}$ 
    & 1 & 1
    &
    $\begin{cases}
    0\\
    -k_1ab+k_2a
    \end{cases}$
    \\ \hline
\end{longtable}

\end{document}